\documentclass[12pt]{article}
\usepackage{amsthm,amsmath,amsfonts,amssymb,cases}
\usepackage{color}

\newcommand{\R}{{\mathord{\mathbb R}}}

\newcommand{\N}{{\mathord{\mathbb N}}}
\newcommand{\C}{{\mathord{\mathbb C}}}

\def\chib {\overline{\chi}}

\newcommand{\HH}{\mathcal{H}}

\newcommand{\FF}{\mathcal{F}}

\newcommand{\WW}{\mathcal{W}}

\newcommand{\UU}{\mathcal{U}}

\newcommand{\ben}{\begin{displaymath}}
\newcommand{\een}{\end{displaymath}}
\newcommand{\beqn}{\begin{equation}}
\newcommand{\eeqn}{\end{equation}}
\newcommand{\beqna}{\begin{eqnarray*}}
\newcommand{\eeqna}{\end{eqnarray*}}

\def\supp{\operatorname{supp}}

\newtheorem{lemma}{Lemma}
\newtheorem{theorem}[lemma]{Theorem}

\numberwithin{equation}{section}
\numberwithin{lemma}{section}

\begin{document}
\title{
Uniqueness of the ground state in the Feshbach renormalization analysis}
\author{\vspace{5pt} D. Hasler$^1$\footnote{
E-mail: dghasler@wm.edu, on leave from Ludwig Maximilians University, Munich.}
 and I.
Herbst$^2$\footnote{E-mail: iwh@virginia.edu.} \\
\vspace{-4pt} \small{$1.$ Department of Mathematics,
College of William \& Mary } \\ \small{Williamsburg, VA }\\
\vspace{-4pt}
\small{$2.$ Department of Mathematics, University of Virginia,} \\
\small{Charlottesville, VA, USA}\\}
\date{}
\maketitle

\begin{abstract}
In the operator theoretic renormalization analysis introduced by Bach,  Fr\"ohlich, and Sigal we prove uniqueness of the ground state.
\end{abstract}

\textbf{Mathematics Subject Classification (2010)}. 81T17 \\

\textbf{keywords}.  Bach-Fr\"ohlich-Sigal, renormalization group, ground state, uniqueness

\section{Introduction}
\label{sec:intro}

In 1998 Bach, Fr\"ohlich, and Sigal (BFS) introduced an operator theoretic renormalization group scheme  \cite{BFS98} to analyze certain field theoretic models.  This scheme was applied in \cite{BFS98} and refined by Bach, Chen, Fr\"ohlich, and Sigal \cite{BCFS03} (see also the work of Griesemer and Hasler \cite{GH08}).  More recently the renormalization group analysis was applied in \cite {GH09, S07, C1, HH10-1, HH10-2, HH10-3, FGS}.   If there is a positivity improving representation of the semigroup generated by the Hamiltonian then a ground state can often be shown to be unique for arbitrary coupling constant (if it exists). There are other methods (for example \cite{BFS98-2}, \cite{S04}) that have been used to show that in some models with regular infrared behavior or minimal coupling, the multiplicity of the ground state is the same as for the electronic Hamiltonian for small coupling. We include the usual model of non-relativistic quantum electrodynamics
(QED) as a model with regular infrared behavior (at least for its behavior near the ground state energy). This occurs because of an approximate gauge invariance.  In the present paper we show that for small coupling, and under rather weak hypotheses, the renormalization group analysis leads to uniqueness of the ground state if the electronic Hamiltonian has a unique ground state even in certain QED models without any gauge invariance.

In Section~\ref{sec:abst} we state the main result in Theorem  \ref{thm:abstract}. In Section~\ref{sec:app} we
apply Theorem~\ref{thm:abstract} to explicit models of non-relativistic QED. In the Appendix we apply
Theorem~\ref{thm:abstract} to operator theoretic renormalization as defined in \cite{BCFS03}.

\section{An Abstract Result}
\label{sec:abst}
We give the necessary hypotheses for an abstract result and then specialize to get to the applications in the next section.  A further explicit application based directly on \cite{BCFS03} is in an Appendix.
Thus suppose $\mathcal H$ is a separable Hilbert space and $T_0$ a self-adjoint operator in $\mathcal H$  with spectrum $\sigma(T_0) = [0,\infty)$ such that $T_0$ has an eigenvalue at $0$ of multiplicity $m$ (thus $m >0$).
Suppose $\{\Gamma_t\}_{t>0}$ is a group of unitary scale transformations on $\mathcal H$ satisfying $\Gamma_s\Gamma_u = \Gamma_{su}$ and that $T_0$ is such that $\Gamma_s T_0 \Gamma_s^{-1}= sT_0$.  We define $\mathcal H_{\rm red}: = 1_{[0,1]}(T_0)\mathcal H$.

We  assume we have a sequence of bounded operators $\{H_n\}_{n=1}^{\infty}$ on $\mathcal H_{\rm red}$, such that  the members of the sequence $\{H_n\}_{n=1}^{\infty}$
are related by the fact that $H_{n+1}$ is up to a constant multiple and a scale transformation a Feshbach transform of $H_n$, that is we assume that  \eqref{eqn:renorm} below holds.
To state this more precisely  we need some further notation.  Suppose $\chi$ and $\chib$ are non-negative Borel functions on ${\R}$ satisfying $\chi^2 + \chib^2 = 1$ with $\chi(x) = 1$ for $x\in [0,a]$ for some $a \in (0,1]$ and $\chi(x) = 0$ for $x\geq 1$.  We assume that $\chi$ is non-increasing. We write $H_n = T_n + W_n$  where both $T_n$ and $W_n$ are bounded  operators on $\mathcal H_{\rm red}$. We assume $T_n$ commutes with $T_0$.  For any $t > 0$ let $\chi_t(x) = \chi(t^{-1}x)$ and similarly for $\chib_t$.  Choose $0 < \rho < a$.  Define
\begin{equation} \label{eqn:1}
\overline H_n: = (T_n + \chib_{\rho}(T_0)W_n\chib_{\rho}(T_0))|_{\textrm {Ran}\chib_{\rho}(T_0)\bigcap\mathcal H_{\rm red}}.
\end{equation}
We assume that $\overline H_n$ and $T_n$ are invertible on Ran$\chib_{\rho}(T_0)\bigcap\mathcal H_{\rm red}$ with bounded inverses and define
\begin{equation}\label{eqn:2}
F_n: = T_n + \chi_{\rho}W_n\chi_{\rho} - \chi_{\rho}W_n\chib_{\rho}(\overline H_n)^{-1}\chib_{\rho}W_n\chi_{\rho}
\end{equation}
where here the arguments of $\chi_{\rho}$ and $\chib_{\rho}$ are both $T_0$. Then we assume
\begin{equation} \label{eqn:renorm}
H_{n+1} = \rho^{-1}\Gamma_{\rho}F_n\Gamma_{\rho}^{-1} .
\end{equation}
This equality is on the subspace $\mathcal H_{\rm red}$.

Note that $\Gamma_{\rho}^{-1}: \mathcal H_{\rm red} \to 1_{[0,\rho]}(T_0)\mathcal H$ and that $F_n$ preserves the latter space.  Finally $\Gamma_{\rho}:1_{[0,\rho]}(T_0) \mathcal H \to \mathcal H_{\rm red}$.
A crucial component of our analysis is the fact shown in \cite{GH08} that in our context $\chi_{\rho}(T_0)$ is an isomorphism of ker$H_n$ onto ker$F_n$ and thus $\Gamma_{\rho}\chi_{\rho}(T_0)$ is an isomorphism of ker$H_n$ onto ker$H_{n+1}$. (It was shown in \cite{BCFS03} that $\chi_{\rho}(T_0)$ is injective as a map from  ker$H_n \to$ {ker}$F_n$ and this is all we will use below.)  It follows that $(\Gamma_{\rho}\chi_{\rho}(T_0))^n$ is an isomorphism of ker$H_1$ onto ker$H_{n+1}$.

\begin{theorem} \label{thm:abstract}

Given the hypotheses
 and definitions above, we make the following additional assumptions.  There exist $ \delta_0 > 0$ and  $\{a_n\}_{n=1}^{\infty} \in \ell^2(\N)$,  such that

\begin{itemize}
\item[(a)] $||T_n \psi|| \geq \delta_0||T_0 \psi||  - a_n||\psi||$ for all $\psi \in \mathcal H_{\rm red}$,
\item[(b)] $||W_n|| \leq a_n$.
\end{itemize}
Then $ { \rm ker} \, H_1$ is at most $m$ dimensional.

\end{theorem}

\begin{proof}  Suppose $ \rm{ker}$$H_1$ is at least $m+1$ dimensional.  Then there exists a non-zero vector $\psi_1$  $\in$ $\rm{ker}$$H_1$ so that $1_{\{0\}}(T_0)\psi_1 = 0$.   It follows that $(\Gamma_{\rho}\chi_{\rho}(T_0))^n\psi_1$ is a non-zero vector in $\rm{ker}$$H_{n+1}$.  Note that because $\rho < a$, $\chi_{\rho}\chi =\chi_{\rho}$.  This gives $(\Gamma_{\rho}\chi_{\rho}(T_0))^n =\Gamma_{\rho^n}\chi_{\rho^n}(T_0)$.  Now it is readily verified that
\begin{equation}
1_{\{0\}}(T_0) + \sum_{j=n}^{\infty}(1_{[a\rho^{j+1},\infty)}(T_0)\chi_{\rho^j}(T_0))^2 \geq \chi_{\rho^n}(T_0)^2.
\end{equation}
Thus
\begin{align} \label{ineq:T_01}
&(a\rho)^2 ||\Gamma_{\rho^n}\chi_{\rho^n}(T_0)\psi_1||^2 =(a\rho)^2 ||\chi_{\rho^n}(T_0)\psi_1||^2 \notag \\
& \leq (a\rho)^2\sum_{j=n}^{\infty}||1_{[a\rho^{j+1},\infty)}(T_0)\chi_{\rho^j}(T_0)\psi_1||^2 \notag \\
& \leq \sum_{j=n}^{\infty}\rho^{-2j}||T_01_{[a\rho^{j+1},\infty)}(T_0)\chi_{\rho^j}(T_0)\psi_1||^2 \notag \\
& \leq \sum_{j=n}^{\infty}\rho^{-2j}||T_0\chi_{\rho^j}(T_0)\psi_1||^2 \notag \\
&=\sum_{j=n}^{\infty}\rho^{-2j}||\Gamma_{\rho^j}T_0\chi_{\rho^j}(T_0)\psi_1||^2 \notag\\
&=\sum_{j=n}^{\infty}||T_0\Gamma_{\rho^j}\chi_{\rho^j}(T_0)\psi_1||^2.
\end{align}

By hypothesis
\begin{equation}
0=||H_{j+1}\Gamma_{\rho^j}\chi_{\rho^j}(T_0)\psi_1|| \geq \delta_0||T_0\Gamma_{\rho^j}\chi_{\rho^j}(T_0)\psi_1|| - 2a_{j+1}||\Gamma_{\rho^j}\chi_{\rho^j}(T_0)\psi_1||. \notag
\end{equation}
Thus
\begin{equation} \label{ineq:T_02}
||T_0\Gamma_{\rho^j}\chi_{\rho^j}(T_0)\psi_1|| \leq 2\delta_0^{-1}a_{j+1}||\Gamma_{\rho^j}\chi_{\rho^j}(T_0)\psi_1||.
\end{equation}
By the monotonicity of $\chi$ if $j\geq n$
\begin{equation} \notag
||\Gamma_{\rho^j}\chi_{\rho^j}(T_0)\psi_1|| = ||\chi_{\rho^j}(T_0)\psi_1|| \leq ||\chi_{\rho^n}(T_0)\psi_1|| = ||\Gamma_{\rho^n}\chi_{\rho^n}(T_0)\psi_1||.
\end{equation}
Thus by (\ref{ineq:T_01}) and (\ref{ineq:T_02})
\begin{equation} \notag
(a\rho)^2 ||\Gamma_{\rho^n}\chi_{\rho^n}(T_0)\psi_1||^2 \leq \sum_{j=n}^{\infty}(2\delta_0^{-1}a_{j+1})^2||\Gamma_{\rho^n}\chi_{\rho^n}(T_0)\psi_1||^2 = d_n||\Gamma_{\rho^n}\chi_{\rho^n}(T_0)\psi_1||^2
\end{equation}
where $d_n: = (2\delta_0^{-1})^2\sum_{j=n}^{\infty}|a_{j+1}|^2 \to 0$ as $n \to \infty$.  Thus for large enough $n,  \Gamma_{\rho^n}\chi_{\rho^n}(T_0)\psi_1 =0$, a contradiction.
\end{proof}

\section{Applications}
\label{sec:app}
In the paper \cite{HH10-1} a version of the BFS renormalization group approach \cite{BCFS03} is developed to prove existence of ground states in certain models.  In many models, after one or two initial Feshbach transformations, one arrives at the situation treated in \cite{HH10-1}, Sections 6--10.  From this point on the presentation applies to any model satisfying the hypotheses.  We used the method in the spin-boson model \cite{HH10-1} and in a model of non-relativistic QED (see \cite{HH10-2}, \cite{HH10-3}) to prove existence of a ground state and analyze some of its properties.

In fact in both of the above models uniqueness is known by methods other than the one presented in this paper.  In the spin-boson model one has a representation where $e^{-tH}$ is positivity improving.  In our QED model this cannot work in the multi-electron case since we assumed that the electrons were fermions.  However, there is a simple method due to Hiroshima and Spohn \cite{HS02} (or see Section 15.3 of \cite{S04}) which gives uniqueness if the Hamiltonian of the model is infrared regular, and in minimally coupled QED, this can be accomplished by a so called Pauli-Fierz gauge transformation.  In addition the method of \cite{BFS98-2} works for minimally coupled QED without any infrared regularization.  We would like to point out that there are models where the latter methods will not work but which have unique ground states nevertheless (which can be seen by applying Theorem \ref{thm:abstract} as we do below).  In particular if in our QED model Hamiltonian, given in \eqref{eq:hamilton1} below, one drops the $(A_{\Lambda}(x_j))^2$ terms to obtain a dipole approximation these methods apparently do not give uniqueness but as mentioned in \cite{HH10-2} the proof given there works with no essential change for this model to give existence of a ground state.  (The main reason that we do not need a Pauli-Fierz transformation to regularize the Hamiltonian in the infra-red is because of the absence of terms of the form $\int w[H_f;k]a(k,\lambda) dk$ and their adjoints in the renormalization group iteration procedure.  Of course in models where these terms have good infrared behavior, the method also works (for example see \cite{S07})).

We consider the following models and then apply Theorem \ref{thm:abstract}.  Define the self adjoint operators
\begin{align}
H^{(1)} =& \sum_{n=1}^N (p_j + gA_{\Lambda}(\beta x_j))^2 + V(x_1,\cdots,x_N) +H_f  \label{eq:hamilton1}   \\
H^{(2)} =& \sum_{n=1}^N (p_j^2 + 2g p_j\cdot A_{\Lambda}(\beta x_j)) + V(x_1,\cdots,x_N) +H_f \nonumber  \\
H^{(3)} =& (\sigma_z +1)\otimes I + I \otimes H_f + g \sigma_x\otimes\phi(f) \nonumber
\end{align}
The operators $H^{(1)}$ and $H^{(2)}$ act in the Hilbert space $\mathcal H= L^2_a(\R^{3N};\mathcal F)$ and $\beta , g \in \R$.  Here $\mathcal F$ is the Fock space for transversally polarized photons and the $a$ indicates anti-symmetry in the $N$ electron coordinates, $x_1,\cdots,x_N$.  The potential $V$ is assumed to be invariant under rotations and permutations of the co-ordinates and Kato small with respect to the electron kinetic energy $\sum_{j=1}^N p_j^2$ where $p_j = -i\nabla_{x_j}$.  By $A_\Lambda(x)$ we denote the quantized vector potential at position $x$  \cite{S04,HH10-2}.
The subscript $\Lambda$  indicates a rotation invariant ultraviolet cut-off.  We do not impose an infrared regularization.  The operator $H_f$ is the kinetic energy of the photons each of which has dispersion relation $\omega(k) = |k|$.  See \cite{HH10-2} and \cite{HH10-3} for details of the model.  It is known that these models have ground states (at least for small $g$).  For $H^{(1)}$ (for all $g$) see \cite{GLL01} and references given there.  Although only $H^{(1)}$ is explicitly treated in \cite{HH10-2} the results given there also hold for $H^{(2)}$ by essentially the same but a slightly simpler proof.
For the spin boson Hamiltonian, $H^{(3)}$, see \cite{HH10-1}. This Hamiltonian acts in $\C^2\otimes \mathcal F$. Here  $\FF$ is
the Fock space over the square integrable functions on $\R^3$.
The operator $ \phi(f) = (a^*(f/\sqrt\omega)) + (a^*(f/\sqrt\omega))^*$ where $f$ is in $L^{\infty}(\R^3) \cap L^2(\R^3)$ and $a^*(\cdot)$ is the usual creation operator.
By  $\sigma_z$ and $\sigma_x$ we denote  the  Pauli matrices. Let $H_{\rm atom} = \sum_{j=1}^N p_j^2 + V(x_1,\cdots,x_N)$ in $L^2_a(\R^{3N})$.   Then  we have the following result.
\begin{theorem}
Suppose $H=H^{(1)}$, $H=H^{(2)}$, or $H=H^{(3)}$ with the definitions given above.  Suppose $H_{\rm atom}$ has a non-degenerate ground state in $L^2_a(\R^{3N})$.  Then there exists $g_0 > 0$ so that if $g \in [-g_0,g_0]$, H has a non-degenerate ground state.
\end{theorem}
\begin{proof}
Referring to \cite{HH10-2} and \cite{HH10-1} after one or two Feshbach transformations the operator $H^{(j)} - E^{(j)}_0$ transforms to an operator unitarily equivalent to an operator $H_1 = T_1 + W_1$ acting in $\mathcal H_{\rm red}: = 1_{[0,1]}(H_f)\mathcal F$.  The null spaces of the operator and its Feshbach transform, $H_1$, are isomorphic.  Here $E^{(j)}_0$ is the ground state energy of $H^{(j)}$. We have omitted superscripts.  $T_1$ is a real $C^1$ function of $H_f$.  Thus referring to Theorem \ref{thm:abstract}, $\mathcal H =  \mathcal F$ and $T_0 = H_f$.  The unitary scale transformation $\Gamma_{\rho}$ leaves the vacuum invariant and acts on the usual creation operators as
$$ \Gamma_{\rho}a^*(k)\Gamma_{\rho}^{-1} = \rho^{-3/2}a^*(\rho^{-1}k)$$
We have omitted polarization indices for $H^{(1)}$ and $H^{(2)}$.
The functions $\chi$ and $\chib$ are assumed to be real and in $C^{\infty}(\R)$ with $\chi$ monotone,
\begin{eqnarray}
\chi^2 + \chib^2 = 1, \quad \supp \chi \subset (-\infty,1] , \quad  {\rm and} \ \ \chi(x) = 1 \ \
\textrm{if}   \ \ x \in [0,3/4].  \label{condchi}
\end{eqnarray}
Thus in Theorem \ref{thm:abstract}, $a=3/4$.
In \cite{HH10-1} ,\! \cite{HH10-2}, \!and \cite{HH10-3}, omitting polarization indices for $H^{(1)}$ and $H^{(2)}$, $W_n$ is given by
 \begin{eqnarray} \nonumber
\lefteqn{  W_n = } \\
&&\sum_{l+m =1}^{\infty}\int  \prod_{j=1}^l a^*(k_j)  w_{l,m}^{(n)}[z_n ; H_f; k_1,\cdots,k_l,\widetilde{k}_1,\cdots,\widetilde{k}_m]
 \prod_{i=1}^m a(\widetilde{k}_j) \frac{dk_1}{|k_1|^{1/2}}\cdots \frac{d \tilde{k}_m}{|\tilde{k}_m|^{1/2}}.\nonumber
 \end{eqnarray}
The kernels $w_{l,m}^{(n)}[z_n; r ; k_1,\cdots,k_l, \tilde{k}_1,\cdots,\tilde{k}_m]$ are such that $\|W_n\| \leq \epsilon_0/ 2^n$
(so that $W_n$ satisfies (b) of Theorem \ref{thm:abstract}). The quantity $z_n$ is a real spectral parameter with $|z_n| < c_1/2^n$.  The operators $H_n$ are defined inductively by \eqref{eqn:renorm} where $\rho$ is chosen suitably small.  $T_n$ is a real $C^1$ function of $H_f$ satisfying (for $r \in [0,1]$)
\begin{align}
|(\partial_r T_n(r) -1)| & \leq \epsilon_0 \label{eq:estonT1} \\
|T_n(0)+ z_n| & \leq \epsilon_0/2^n \label{eq:estonT2}
\end{align}
Here $\epsilon_0$ is a small positive number.
It follows that
\begin{equation} \label{eq:estonT3}
T_n(r) = (r + \int_0^r(\partial_s T_n(s) - 1)d s)) +(T_n(0) + z_n) - z_n
\end{equation}
so that
\begin{equation} \label{eq:estonT4}
 \|T_n(H_f)\psi\| \geq (1 - \epsilon_0)\|H_f \psi\| - 2\epsilon_0 2^{-n}\|\psi\|
 \end{equation}
Thus $T_n$ satisfies hypothesis (a) of Theorem \ref{thm:abstract}.  The smallness of $\epsilon_0$ and the bounds satisfied by $T_n$ and $W_n$ require $|g|$ to be small and the proof of these bounds requires $\rho$ to be small, in particular $\rho < 3/4$.

In view of the conformal transformation which is part of the renormalization group analysis of \cite{HH10-1} we make some remarks to justify \eqref{eqn:renorm}. In \cite{HH10-1} operator functions of a complex variable, $\widetilde H_n(z) = \widetilde T_n(z) + \widetilde W_n(z)$ are a crucial part of the analysis and the Hamiltonians $H_n$ are given by evaluating  $\widetilde H_n(z)$ at the real spectral parameter $z_n$, $H_n = \widetilde H_n(z_n)$ where $z_n$ is given by a certain limit.  We have
\begin{equation}\label{eqn:compositionlaw}
z_{n} = J_n^{-1}(z_{n+1})
\end{equation}
where $J_n$ is a conformal transformation (see \cite{HH10-1} for further definitions and details). We have
$$
\widetilde H_{n+1}(J_n(z)) = \rho^{-1}\Gamma_{\rho}\widetilde F_n(z)\Gamma_{\rho}^{-1}
$$
where $\widetilde F_n(z)$ is given by \eqref{eqn:1} and \eqref{eqn:2} with $H_n$ replaced by $\widetilde H_n(z)$ and similarly for $T_n$ and $W_n$.  Substituting $z=z_n$ and noting the composition law \eqref{eqn:compositionlaw}, we see that $H_{n+1}$ and $H_n$ are indeed related by \eqref{eqn:renorm}.

\end{proof}

\appendix

\section{Appendix}
\label{sec:appendix}

In this appendix we show that  Theorem   \ref{thm:abstract} can be used to obtain uniqueness in the framework of operator theoretic renormalization as defined in  \cite{BCFS03}.
To this end we recall the setup  of  \cite{BCFS03}. Let $I=[0,1]$,
 let $B_1$ denote the unit ball in  $\R^d$, and let $D_{1/2} := \{ z \in \C | |z| \leq  1/2 \}$.
 Here the Fock space $\FF$ is over the space of square integrable functions on $\R^d$.
 Let $\WW_{m,n}^\#$ denote the Banach space of
functions $B_1^m\times B^n_1 \to C^1(I)$, which
are a.e. defined, totally symmetric under the interchange of components in $B_1^m$ respectively  $B_1^n$, and which satisfy
the norm bound
$$
\| w_{m,n} \|_\mu^\# := \| w_{m,n} \|_\mu + \| \partial_r w_{m,n} \|_\mu < \infty  ,
$$
where
$$
\| w_{m,n} \|_\mu := \left[  \int_{B_1^{m+n}}
\sup_{r \in I} \left| w_{m,n}[r ; k_1,...,k_m, \tilde{k}_1,..., \tilde{k}_n ] \right|^2 \prod_{i=1}^m \frac{d^d k_i}{|k_i|^{3 + 2 \mu}}
\prod_{j=1}^m \frac{d^d \widetilde{k}_j}{|\widetilde{k}_j|^{3 + 2 \mu}} \right]^{1/2}
$$
for some $\mu > 0$. For $0 < \xi < 1$ one defines, $\WW_\xi^\# := \bigoplus_{m+n \geq 0}  \WW_{m,n}^\#$, to be the Banach space of sequences $\underline{w} = (w_{m,n})_{m+n \geq 1}$
obeying
$$
\| \underline{w} \|_{\mu,\xi}^\# := \sum_{m+n \geq 1} \xi^{-(m+n)} \| w_{m,n} \|_\mu^\# < \infty  .
$$
We define
\begin{equation*}
\begin{split}  W_{m,n}[w_{m,n}] :=  & {1}_{I}(H_f) \int_{B_1^{m+n}} \prod_{j=1}^m  a^*(k_j)
w_{m,n}[H_f; k_1,...,k_m, \widetilde{k}_1,..., \widetilde{k}_n ]   \prod_{l=1}^n  a(\widetilde{k}_l) \cr
& \times  \prod_{j=1}^m \frac{d^d k_j}{|k_j|^{1/2}}     \prod_{l=1}^n \frac{d^d \widetilde{k}_l}{|\widetilde{k}_l|^{1/2}}       \          1_{I}(H_f)      ,
\end{split}
\end{equation*}
where $a^*(k)$ and $a(k)$ denote the usual creation and annihilation operator.
In  \cite{BCFS03} it is shown that $H(\underline{w}) = \sum_{m+n \geq 0} W_{m,n}[w_{m,n}]$ defines
a bounded operator on $\HH_{\rm red} := 1_{I}(H_f)\FF$, with bound
\begin{eqnarray} \label{eq:boundonhw}
 \| H(\underline{w})  \| \leq \| \underline{w} \|_{\mu,\xi}^\# .
\end{eqnarray}
We define   the polydics $\mathcal{D}(\epsilon, \delta )$ to consist of the analytic\footnote{Analytic on the closed set $D_{1/2}$
 means that the function is analytic in an open  neighborhood of $D_{1/2}$.}
 functions
$\underline{w}[\cdot ] : D_{1/2} \to \WW_{\xi}^\#$, with
\begin{align*}
&\sup_{z \in D_{1/2}} \sup_{ r \in I} | \partial_r w_{0,0}[z;r] - 1 | \leq \epsilon , \\
&\sup_{z \in D_{1/2}} | w_{0,0}[z;0]   +  z | \leq \delta  , \quad   \sup_{z \in D_{1/2}} \| (w_{m,n}[z])_{m+n \geq 1} \|^\#_{\mu,\xi} \leq \delta   .
\end{align*}
Let $0 < \rho < 1/2$, and assume $\underline{w} \in \mathcal{D}(\rho/8,\rho/8)$.
In \cite{BCFS03} it  is shown that  the  following map
is biholomorphic
$$
E_{\rho}[\cdot] : \UU[w] \to D_{1/2} , \quad z \mapsto   - \rho^{-1}w_{0,0}[z;0] ,
$$
where $\UU[w] := \{ z \in D_{1/2} | | w_{0,0}[z;0] | \leq \rho/2 \}$.
Assume  $\chi, \chib \in C^\infty(\R;[0,1])$  satisfy \eqref{condchi} and that $\chi$ is monotone.  Define $\chi_\rho := \chi(\rho^{-1} H_f )$ and
 $\chib_\rho := \chib(\rho^{-1} H_f )$.
Let  $\zeta \in D_{1/2}$, and set   $$T(r) := {w}_{0,0}[E_\rho^{-1}(\zeta);r ], \qquad
W := \sum_{m+n \geq 1} W_{m,n}\left[{w}_{m,n}[E_\rho^{-1}(\zeta)]\right].$$ Then, as shown in  \cite{BCFS03},
the map  $H_{\chib_\rho} := T(H_f) + \chib_\rho W \chib_\rho $ is bounded invertible on the range of $\chib_\rho$
and there exists a unique so called renormalized kernel,
$\mathcal{R}_\rho(\underline{w})$,  such that
\begin{equation} \label{eq:defofrg}
H ( \mathcal{R}_\rho(w)[\zeta]) = \frac{1}{\rho} \Gamma_\rho \left(  T(H_f) + \chi_\rho W \chi_\rho - \chi_\rho W \chib_\rho H_{\chib_\rho}^{-1} \chib_\rho W \chi_\rho  \right) \Gamma_\rho^* .
\end{equation}
Furthermore, for fixed $\mu > 0$, there exist  $\rho, \xi, \epsilon_0 > 0$ such that
\begin{equation} \label{eq:contr}
\mathcal{R}_\rho : \mathcal{D}(\epsilon, \delta) \to \mathcal{D}(\epsilon + \delta/2, \delta/2 )
\end{equation}
for all $ \epsilon, \delta \in [0, \epsilon_0]$.
The existence part of the following theorem has been proven in   \cite[Theorem 3.12]{BCFS03}, i.e.,
with  assertion  l.h.s. \eqref{eq:equality} $\geq$ r.h.s.  \eqref{eq:equality}. Using Theorem \ref{thm:abstract}
one can show that one has uniqueness, i.e.,  l.h.s. \eqref{eq:equality} $\leq$ r.h.s.  \eqref{eq:equality}.
\begin{theorem} \label{thm:bcfs}
Fix $\mu > 0$, and choose $\rho, \xi, \epsilon_0 > 0$ sufficiently small such that \eqref{eq:contr}
holds for all $0 \leq \epsilon \leq \epsilon_0$ and $0 \leq \delta \leq \epsilon_0$. Suppose that $\underline{w} \in \mathcal{D}(\epsilon_0/2,\epsilon_0/2)$. Then
the complex number $e_{(0,\infty)} \in D_{1/2}$ defined in  \eqref{eq:defofenergy} is an eigenvalue of $H(\underline{w})$, in the sense that
\begin{equation} \label{eq:equality}
{\rm dim} \ker \left\{ H(\underline{w}[e_{(0,\infty)}])  \right\} = 1  .
\end{equation}
\end{theorem}
\begin{proof} In view of  \cite[Theorem 3.12]{BCFS03} it remains to show  $ {\rm dim} \ker \left\{ H(\underline{w}[e_{(0,\infty)}] \right\} \leq 1 $.
To this end we combine the  proof given in \cite{BCFS03} with
 Theorem  \ref{thm:abstract}.
The contraction property  \eqref{eq:contr} allows  to iterate the renormalization transformation, and hence for $n \in \N_0$ the kernels
$ \underline{w}^{(n)} := \mathcal{R}_\rho^n(\underline{w})$ satisfy
\begin{equation}  \label{eq:contr2}
 \underline{w}^{(n)}  \in \mathcal{D}(\epsilon_0 , 2^{-n} \epsilon_0) .
\end{equation}
We define $E_n[z] := w^{(n)}_{0,0}[z;0]$, $J_n[z] = \rho^{-1} E_n[z]$, and for $n  \leq m \in \N_0$
$$
e_{(n,m)} := J^{-1}_{(n)} \circ \cdots \circ J^{-1}_{(m)}[0] .
$$
Using  property  \eqref{eq:contr2} it was shown in \cite{BCFS03} (c.f. Eq. (3.146)) that the following limit exits
\begin{equation} \label{eq:defofenergy}
z_n := e_{(n,\infty)} := \lim_{m \to \infty} e_{(n,m)} .
\end{equation}
It follows from the definition of the renormalization transformation  \eqref{eq:defofrg} that
$$
H_{n}  =  \frac{1}{\rho} \Gamma_\rho  F_{n-1} \Gamma_\rho^* ,
$$
where
\begin{eqnarray*}
H_{n} & := & H(\underline{w}^{(n)}[e_{(n,\infty)} ]) ,   \\
T_n(r) & := & w^{(n)}_{0,0}[e_{(n,\infty)};r]  , \quad
W_n  :=   \sum_{m+l \geq 1} W_{m,l}\big[{w}^{(n)}_{m,l}[e_{(n,\infty)} ] \big]     , \\
F_{n} & := &       T_n(H_f)  + \chi_\rho W_n \chi_\rho - \chi_\rho W_n \chib_\rho \left( T_n(H_f) + \chib_\rho W_n \chib_\rho  \right)^{-1} \chib_\rho W_n \chi_\rho          .
\end{eqnarray*}
It follows from \eqref{eq:contr2}, that  \eqref{eq:estonT1} and  \eqref{eq:estonT2} hold. Then using  the decomposition
\eqref{eq:estonT3} we find  \eqref{eq:estonT4}.
Thus $T_n(H_f)$ satisfies Hypothesis (a) of Theorem \ref{thm:abstract} with $T_0 =H_f$. Now Hypthesis (b) of Theorem \ref{thm:abstract} follows from
\begin{equation} \label{eq:estonWn}
\| W_n  \| \leq \epsilon_0 2^{-n} .
\end{equation}
Ineq.  \eqref{eq:estonWn}   can be shown using    \eqref{eq:boundonhw} and  \eqref{eq:contr2}.
 The theorem now follows as a consequence of    Theorem  \ref{thm:abstract}.
\end{proof}

\bibliography{qed}

\bibliographystyle{plain}

\end{document}